\documentclass[journal]{IEEEtran}

\newcommand{\ls}[1]     
{\dimen0=\fontdimen6\the=#1\dimen0
	\advance\lineskip.5\fontdimen5\the\lineskip-\dimen0
	\lineskiplimit=.9\lineskip
	\baselineskip=\lineskip     \advance\baselineskip\dimen0
	\normallineskip\lineskip
	\normallineskiplimit\lineskiplimit
	\normalbaselineskip\baselineskip
	\ignorespaces
}
\newcommand{\vs}{\vspace{-0.08in}}

\usepackage{graphicx}
\graphicspath{{figures/}}
\DeclareGraphicsExtensions{.eps}
\usepackage{epstopdf}
\usepackage{algorithm}
\usepackage[tight]{subfigure}
\usepackage{caption}
\usepackage{tabularx}
\usepackage{calc}
\usepackage{enumitem}
\usepackage{mathtools}
\usepackage{mathrsfs}
\usepackage{multirow}
\usepackage{algpseudocode}
\usepackage{array}
\usepackage{setspace}
\usepackage{makecell}
\usepackage{amsthm}
\usepackage{url}

\algtext*{EndWhile}
\algtext*{EndIf}
\algtext*{EndFor}

\newtheorem{lemma}{Lemma}
\newtheorem{theorem}{Theorem}

\ls{1}

\begin{document}
	

	\title{Enabling Failure-resilient Intermittent Systems Without Runtime Checkpointing}

	\author{Wei-Ming Chen,~\IEEEmembership{Student Member,~IEEE},~Tei-Wei-Kuo,~\IEEEmembership{Fellow,~IEEE},~and Pi-Cheng Hsiu,~\IEEEmembership{Senior Member,~IEEE}
	\thanks{W.-M. Chen is with the Department of Computer Science and Information Engineering, National Taiwan University, No. 1, Sec. 4, Roosevelt Rd., Taipei 10617, Taiwan, and also with Research Center for Information Technology Innovation (CITI), Academia Sinica, No. 128, Sec. 2, Academia Rd., Nankang Dist., Taipei 115, Taiwan (E-mail: D04922006@csie.ntu.edu.tw).}
	\thanks{T.-W. Kuo is with the Department of Computer Science and Information Engineering, National Taiwan University, Taipei 106, Taiwan, and also with College of Engineering, City University of Hong Kong, 88 Tat Chee Avenue, Kowloon Tong, Hong Kong (Email: ktw@csie.ntu.edu.tw).}
	\thanks{P.-C. Hsiu is with the Research Center for Information Technology Innovation (CITI), and the Institute of Information Science (IIS), Academia Sinica, No. 128, Sec. 2, Academia Rd., Nankang Dist., Taipei 115, Taiwan, and also with the Department of Computer Science and Information Engineering, National Chi Nan University, No. 1, University Rd., Puli, Nantou 54561, Taiwan (E-mail: pchsiu@citi.sinica.edu.tw).}
	\thanks{A preliminary version of this paper was presented at the IEEE/ACM Design Automation Conference (DAC) 2019.}}

	\maketitle
	
	\ls{0.92} 
	
	\begin{abstract}
		
		Self-powered intermittent systems typically adopt runtime checkpointing as a means to accumulate computation progress across power cycles and recover system status from power failures. However, existing approaches based on the checkpointing paradigm normally require system suspension and/or logging at runtime. This paper presents a design which overcomes the drawbacks of checkpointing-based approaches, to enable failure-resilient intermittent systems. Our design allows accumulative execution and instant system recovery under frequent power failures while enforcing the serializability of concurrent task execution to improve computation progress and ensuring data consistency without system suspension during runtime, by leveraging the characteristics of data accessed in hybrid memory. We integrated the design into FreeRTOS running on a Texas Instruments device. Experimental results show that our design can still accumulate progress when the power source is too weak for checkpointing-based approaches to make progress, and improves the computation progress by up to 43\% under a relatively strong power source, while reducing the recovery time by at least 90\%.
		

	\end{abstract}
	

	\section{Introduction} \label{sec:Introduction}
	Applications based on smart embedded devices have become a ubiquitous part of daily life. However, powering such devices is a critical challenge because of their size restrictions and applications in large-scale scenarios. Energy harvesting has emerged as a promising alternative power source for these devices. To enable \emph{intermittent computing}, self-powered systems typically checkpoint execution progress and data residing in volatile memory (VM) to non-volatile memory (NVM) at runtime such that the systems can be recovered after power resumption. However, because ambient power sources suffer from frequent power failures, the overheads incurred by frequent checkpoints could significantly reduce system performance, thus increasing the difficulty of designing hardware chips and system software.

	Many attempts have been made to enable intermittent systems, which can survive in unstable power environments, at the level of hardware circuits, system architectures, and system software by efficiently checkpointing data residing in VM to NVM~\cite{7750989,Liu2747910,8519610}. To accumulate execution progress made in different power-on periods, \emph{non-volatile processors} (NVPs) have emerged as a potential solution by \emph{checkpointing} volatile states in the CPU registers, allowing the system to resume from the drop off point when power is restored~\cite{6341281}.
	The volatile states in main memory, including data, stacks, and heaps of tasks, can also be backed up to non-volatile memory so that the entire system can be recovered by restoring the checkpointed states after power resumption~\cite{6733152}. Various mechanisms based on the checkpointing paradigm have also been introduced to adapt peripheral I/O devices (e.g., sensors~\cite{7544395}, Wi-Fi modules~\cite{8824923}, and electrophoretic displays~\cite{EDisplay}) to intermittent power supply~\cite{8585137}.
	
	Recently, increased interest has focused on adapting system software to NVP-based devices. The compilers for NVP-based systems have been designed to reduce the size of checkpointing data (e.g., stack~\cite{7167369} and register~\cite{6865319}), thus increasing checkpointing efficiency. In addition, task schedulers have been investigated to improve quality delivered by the system in terms of respective performance indexes (e.g., the deadline miss rate~\cite{7167310} or system value~\cite{7809860}) under different application scenarios. To optimize performance in a best-effort fashion with unpredictable power supply, the \emph{forward progress} of intermittent task execution was deemed a sensible index and maximized using redesigned resource allocation policies (e.g., the scheduler~\cite{Pan3081038} or power manager~\cite{Ma3077575}).
	Under weak power supply, program sections may be longer than power-on periods. Thus, to ensure forward progress while avoiding repeated code execution, \emph{program atomicity} was supported in~\cite{ISLPED2018Kang} by ensuring that an uninterruptible code section can be run through at one execution, and \emph{progress stagnation} has been addressed in~\cite{RTAS2019} by dynamically adapting the checkpoint interval and size to the harvested energy.
	
	Without careful consideration of different system snapshots in the memory hierarchy, the checkpointing paradigm may suffer from inconsistency between the data in non-volatile memory and the restored task progress~\cite{Ransford:2014:NMB:2618128.2618136}.
	How to achieve \emph{data consistency} is a critical issue because correctness is one of the basic requirements of computer systems. Some solutions have been proposed to eliminate consistency errors. In particular, consistency-aware checkpointing approaches have been proposed to checkpoint the system at safe lines of program code~\cite{ConsistencyAware} or to insert auxiliary code to ensure the correctness of all checkpoints~\cite{Xie:2018:ADI:3184476.3182170}. A hardware scheme has been proposed to automatically checkpoint system states while discarding all speculative modifications which may lead to inconsistency~\cite{7547183}. Moreover, programming models have been proposed to prevent errors by performing \emph{data versioning} for non-volatile data~\cite{Lucia2737978} or using a task-based execution model which only allows executing one task at a time in the system~\cite{Maeng:2017:AIE:3152284.3133920}.
	However, these solutions are either based on the checkpointing paradigm which requires system suspension to backup volatile data frequently, resulting in non-negligible runtime overheads, or require changing the existing programming model, imposing a burden on application developers.

	This paper proposes a failure-resilient design which overcomes the drawbacks of checkpointing-based approaches while preserving progress across power cycles. Our design, which is compatible with multitasking operating systems, enables intermittent systems to (1) run multiple tasks \emph{concurrently} to improve computation progress, (2) achieve data \emph{consistency} without system suspension during runtime, (3) recover \emph{instantly} from power failures, and (4) \emph{accumulatively} preserve computation progress across power cycles to avoid stagnation. To realize the design, we add a \textit{data manager} and a \textit{recovery handler} in an operating system, so that the system runtime can cope with intermittence and exempts application developers from this responsibility. The idea behind the design is to leverage the characteristics of data accessed in hybrid memory, where VM provides high-performance data access while NVM provides data persistency when power failures occur.

	
	
	However, endowing intermittent systems with the four abilities raises corresponding challenges. First, \emph{serializability} of concurrent task execution must be guaranteed. In our design, the data manager allows two-version copies for each data object in VM and NVM to increase the concurrency, while ensuring that data objects modified by tasks in VM are written into NVM \emph{atomically} and will not violate the serializability of task execution. Second, \emph{data consistency} must be maintained. To this end, the recovery handler tracks the progress of all tasks, while the data manager ensures that a nonvolatile data version in NVM is consistent with the progress of \emph{finished} tasks at all times. This also guarantees that a persistently consistent version is always available in NVM. Third, because the data is instantly recoverable, after power resumption, the recovery handler only needs to recreate and rerun \emph{unfinished} tasks in VM, thereby achieving instant system recovery. Finally, to prevent tasks whose execution times are longer than power-on periods from being repeatedly recreated and rerun, the data manager will allocate their contexts in NVM so that the recovery handler can accumulatively complete them across multiple power cycles.

	
	To evaluate the efficacy of our design, we integrated it into a real-time operating system called FreeRTOS, and conducted extensive experiments with a collection of real tasks on an ultra-lightweight platform, namely the Texas Instruments MSP-EXP430FR5994 LaunchPad. Compared to checkpointing-based approaches that require runtime checkpointing and/or system logging~\cite{6733152,databaselogging}, the proposed design can improve the forward progress by between 8\% and 49\% while maintaining data consistency under a strong power source, and by infinity when checkpointing-based approaches cannot make progress under a weak power source. Experimental results with various power traces also show that our design can reduce recovery time by at least 90\%, thus making it particularly suitable for self-powered devices which may suffer from frequent power failures.


	The remainder of the paper is organized as follows. Section~\ref{sec:Backgroundand} provides background information and explains some drawbacks of existing approaches based on checkpointing. In Section~\ref{sec:design}, we present the details of our failure-resilient design. Experimental results are reported in Section~\ref{sec:Evaluation}. Section~\ref{sec:conclusion} presents some concluding remarks.

	\section{Background and Motivation} \label{sec:Backgroundand}

	\subsection{Ultra-lightweight Intermittent Devices}
	\begin{figure}[h]	
		\centering
		\includegraphics[width=0.9\columnwidth]{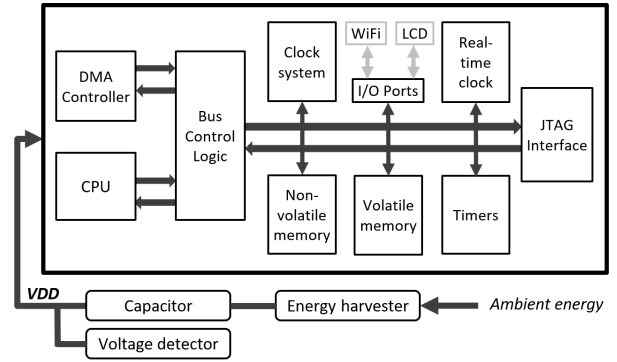}
		\caption{System architecture of a self-powered device.}\label{fig:System}
	\end{figure}
	
	\subsubsection{Hardware Architecture}
	
	Figure \ref{fig:System} shows the system architecture of a typical ultra-lightweight device equipped with various hardware components. To provide basic computing functionality, such devices must contain essential hardware components like a CPU and main memory. The CPU executes program code in memory and performs general logic and arithmetic operations on data in the processor's registers and memory. Recently, such devices have increasingly used hybrid memory architectures to take advantage of the characteristics of hybrid memory. Specifically, volatile memory (VM) features high performance and low energy consumption for data access and is usually used to store runtime data, like task stacks and intermediate results. By contrast, non-volatile memory (NVM) features non-volatility and high capacity and is usually used to preserve data when power failures occur.
	To provide additional functionality that may be required by various applications, an ultra-lightweight device can be equipped with extra components like a DMA controller, a timer, a system clock, and external I/O ports.
	The DMA controller allows applications to manage memory without occupying CPU time, in that external hardware components connected via the I/O ports can directly access data in main memory by the DMA controller.
	For applications requiring timely responses, the real-time clock and timer can be used to measure time and trigger interrupt functions to handle events in real time.
	
	To provide mobility without frequent recharging, energy harvesting has emerged as a promising power source for ultra-lightweight devices. However, power supplies reliant on energy harvesting are inherently unpredictable and unstable, increasing the difficulty of designing intermittently-powered devices. For example, a sudden power loss will cause unsaved volatile data and the computing progress of tasks to be lost.
	To deal with this issue, non-volatile processors (NVP) have emerged as a promising alternative to traditional processors. An intermittently-powered device equipped with an NVP typically contains a voltage detector and a capacitor. The capacitor saves (resp. uses) additional energy when the input (resp. output) voltage is higher (resp. lower) than the output (resp. input) voltage, while the voltage detector monitors the power supply voltage and can trigger specific functions when the voltage falls to a predefined threshold.
	By implementing backup/restore mechanisms triggered by the voltage detector, several checkpointing-based solutions have been proposed to allow for intermittent task execution~\cite{6733152,SuddenPower}.
	For example, an intuitive approach is to checkpoint all volatile data in VM (including registers and main memory) to NVM when the voltage falls below a threshold and then write the data back to VM when power is resumed.
	
	\subsubsection{System Software}
	
	A lightweight operating system, providing system services and exempting application developers from the responsibility of managing hardware resources, usually employs a scheduler to support \emph{multitasking} and control the execution order of tasks. Specifically, once the system boots up, the scheduler will setup the timer to generate periodic interrupts that divide CPU time into slices. Whenever the CPU handles an interrupt, the scheduler is invoked to allocate the next time slice to a task selected to occupy the CPU and access memory in the subsequent time slice. Note that, if the selected task differs from the currently running task, the scheduler will first perform \textit{context switch}, which saves the running task's context by pushing the data of the CPU registers into the running task's stack in memory and then restores the selected task's context by popping the previously saved data in the selected task's stack into the CPU's registers.
	
	
	In a multitasking operating system, which enables tasks to be executed in an interleaving manner, the system typically supports \textit{concurrency control} to allow concurrently executed tasks to access shared data objects. When tasks attempt to access the same data objects via the provided data access operations, the operating system controls the order of data access operations invoked by the tasks and manipulates the copies of data objects in memory, keeping data management being transparent to the tasks.
	However, when data objects are concurrently accessed by interleavingly executed tasks, the outcome of data objects is not deterministic and depends on the execution order of the operations invoked by the tasks. To ensure data access predictability, the operating system should ensure that each task can be deemed to be executed in isolation by guaranteeing \emph{serializability}, in that the concurrent execution of tasks must be equivalent to the case where these tasks are executed serially in \emph{some} arbitrary order. Note that any serial order of task execution is legitimate, so the resultant values of data are not deterministic. By allowing more tasks to be executed concurrently, the operating system increases the CPU utilization and thus improves the forward progress achieved by the system.
	
	

	\subsection{Drawbacks of Checkpointing}\label{ssec:example}
	
	To preserve the forward progress of task execution, typical intermittent systems have to frequently checkpoint task status and/or data at runtime. However, adopting checkpointing-based approaches in intermittently-powered devices presents some critical drawbacks.
	First, to preserve execution progress made between power failures, at runtime these approaches periodically checkpoint the (entire or partial) snapshot in VM to NVM, so that, after power resumption, the system can be recovered to the latest checkpoint by restoring the snapshot from NVM to VM. Consequently, \emph{data inconsistency} may occur if some data in NVM is modified between the latest checkpoint and a power failure~\cite{7753544}. Specifically, after power resumption, the execution progress will be rolled back to the latest checkpoint, whereas the data in NVM cannot be rolled back. This leads to data inconsistency between VM and NVM because the data in NVM may be modified again.
	

	To achieve data consistency, a straightforward approach is to adopt \textit{system-wise checkpointing}, which checkpoints an entire system snapshot, including data, heaps, and stacks of tasks~\cite{6733152}. This approach requires a lengthy suspension of all running tasks to ensure that all volatile content in VM is exclusively accessed by the checkpointing procedure, resulting in extra runtime overhead. To reduce the checkpoint size and time required by checkpointing, an alternative approach is to adopt \textit{logging-based checkpointing}, which records and dumps all write-ahead logs and modified data residing in VM to NVM~\cite{databaselogging}. In this way, the system can traverse logs to recover inconsistent data accordingly by redoing (resp. undoing) modifications made by finished (resp. unfinished) tasks. However, such logging-based checkpointing approaches suffer from long recovery time due to log traversing and progress loss of unfinished tasks whenever a power failure occurs. For intermittently-powered devices, which could suffer from extremely frequent power failures, the checkpointing paradigm may be unable to provide timely checkpointing and data recovery based on logs within a short power-on period.
	
	This observation suggests that intermittently-powered devices should be capable of not only \emph{progress accumulation} within short power-on periods but also \emph{instant recovery} immediately after power resumption. Furthermore, to improve the forward progress, \emph{task concurrency} and \emph{data consistency} should be achieved without runtime suspension and logging.

	\section{Failure-resilient Task Execution}  \label{sec:design}
	
	In this section, we present a failure-resilient design which allows instant system recovery from power failures and enables computation progress accumulation while achieving data consistency and the serializability of concurrent task execution without runtime checkpointing and logging. The rationale behind our design is to ensure that tasks are executed \emph{serially} in the logical sense and all modifications to data objects in NVM are written \emph{atomically}, while computation progress is accumulated by allocating data in VM or NVM instead of copying data from VM to NVM.
	Two components, namely a data manager and a recovery handler, are developed to realize the design. Section~\ref{ssec:overview} gives a design overview, while Sections~\ref{ssec:data} and~\ref{ssec:recovery} respectively present some design details of the data manger and the recovery handler.

	\subsection{Design Overview} \label{ssec:overview}
	
	\begin{figure*}[h]\vspace{0.1in}
		\centering
		\vs
		\includegraphics[width=1.9\columnwidth]{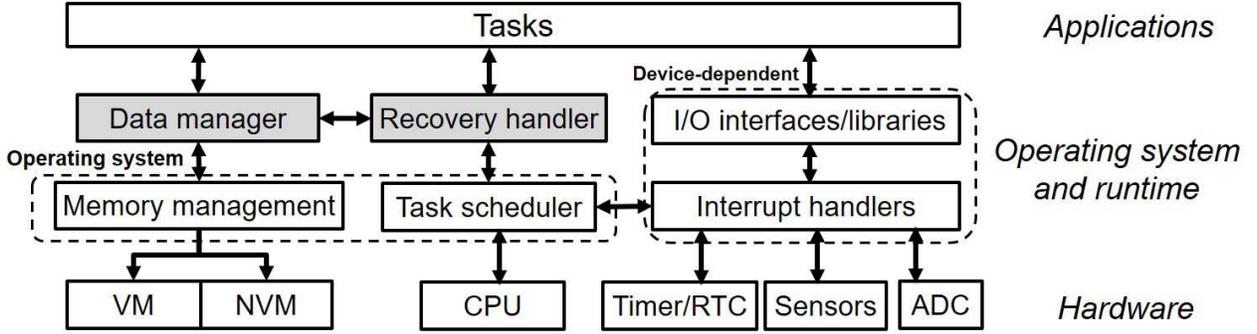}\vspace{-0in}
		\caption{Our failure-resilient design.}\label{fig:SystemOverview}
	\end{figure*}

	As shown in Figure~\ref{fig:SystemOverview}, a lightweight operating system typically provides a task scheduler and memory management to support multitasking and allow tasks to access data in memory. The scheduler provides functions to create and delete tasks and controls the execution order of tasks. Whenever a task is created, the scheduler will allocate memory space in VM by default to the task and initialize the task's status. Then, the task will enter a ready queue and wait to be scheduled. At runtime, concurrently executed tasks can \textit{read}, \textit{write}, and \textit{commit} data objects through the corresponding data access operations provided by the operating system. If data access operations made by interleavingly executed tasks are uncontrolled, the resultant values of data are unpredictable and may be undesirable. Therefore, the operating system should guarantee \emph{serializability}, in that the outcome of concurrently executed tasks is equivalent to the outcome of serially executed tasks in \emph{any} serial order. The adoption of task concurrency can improve forward progress but significantly complicates data management in intermittent systems.
	Specifically, consistency between the data and execution progress of tasks must be achieved. This is particularly difficult for lengthy tasks whose execution times are longer than power-on periods, because a lengthy task can finish only if its execution progress is accumulated across different power-on periods; otherwise, it may continuously rerun and never finish.

	Our design enables intermittently-powered systems to be capable of failure-resilient task concurrency without runtime checkpointing and logging. As shown in Figure~\ref{fig:SystemOverview}, we employ a data manager to enforce the atomicity and serializability of concurrent task execution while maximizing forward progress, as well as a recovery handler to instantly recover the system after power is resumed.
	The data manager is responsible for allocating and maintaining data and task status in VM and NVM. To ensure serializability, it replaces the original implementations of read, write, and commit operations, and allows two-version copies respectively in VM and NVM for each data object. Moreover, the data manager monitors the operations invoked by every task and validates whether serializability will be violated immediately before the task attempts to commit its modifications to data copies from VM to NVM. If the serializability is violated, the task is simply aborted and recreated. To maintain data consistency, the recovery handler is responsible for keeping track of task execution progress as tasks are created, finished, and aborted. Specifically, once a task is created by the scheduler, the recovery handler records the task's attributes in NVM so that all unfinished tasks, which are volatile in VM, can be recreated after power resumption or task abortion. After a task is finished by successfully committing its modifications to data objects from VM to NVM, the recovery handler marks the task as finished, preventing the committed data objects from being inconsistent due to repeatedly modified by finished tasks.
	
	The data manager and the recovery handler also cooperate to accumulate progress of lengthy tasks whose execution times are too long to be finished within one power-on period. Specifically, after power resumption, the recovery handler determines whether a task is lengthy based on
	whether the task has ever been rerun due to a power failure. Once a task is deemed lengthy while being recreated by the recovery handler, the data manager will allocate memory space in NVM (instead the default VM) to the task so that its execution progress will become nonvolatile at the cost of lower execution performance. To avoid data inconsistency, before a power failure occurs (detected by a voltage detector in our implementation), the recovery handler enforces the scheduler to context switch the currently executed lengthy task (if any) to prevent it from being scheduled at a low voltage. After the task is switched out, the data of the CPU registers are automatically pushed to the top of its stack and its context can be preserved in NVM during power-off periods. After power resumption, those unfinished lengthy tasks can instantly resume from where they left off by simply being added into the scheduler during system recovery. This design allows for lengthy tasks to accumulate their progress across power cycles without additional overhead of memory copying required by runtime checkpointing between VM and NVM. Note that, to ensure serializability, the data manager allocates data copies modified by lengthy tasks in NVM and performs serializability validation as usual before a lengthy task attempts to commit its modifications. If serializability is violated, the lengthy task is also aborted and recreated.

	\subsection{Consistency-aware Memory Management}\label{ssec:data}
	
	\subsubsection{Task Context Allocation}
	
	After a task is created by the scheduler, the data manager maintains the memory space allocated to the task as well as its stack. A task's stack stores local variables created by unfinished function calls invoked by the task. These variables will be declared and initialized by the system and then modified by the task at runtime. To maximize computation efficiency, when a task is created, its stack is allocated in VM by default and, during task execution, some variables in the stack will be fetched into the CPU registers. Because the stack size is usually fixed and needs to be specified prior to task creation, the operating system normally supports dynamic memory allocation as well, allowing a task to acquire additional memory space to store local variables whose sizes will be specified at runtime. The data manager allocates the additional memory space required by the task from the system \emph{heap} via system calls (e.g., \texttt{malloc()} and \texttt{free()} in a system supported standard C library). Because the stacks and heaps of tasks are allocated in VM by default, when a power failure occurs, the contexts of unfinished tasks, which have yet to commit the modified data to NVM, will be lost as if the tasks have never been executed.
	
	To preserve the contexts of lengthy tasks (as determined by the recovery handler) and modified data during power-off periods, the data manager allocates their stacks, heaps, and all used memory space in NVM instead. Moreover, the data manager uses the memory management mechanism provided by the operating system to maintain the memory space used by lengthy tasks in some data structures, which are stored in NVM so that, after power resumption, the stacks, heaps, and memory space allocated to lengthy tasks can be found and reused accordingly. However, if a power failure occurs during the execution of a lengthy task, its context will become invalid because the variables currently fetched into the CPU registers will be lost, whereas the stack and heap will still be preserved in NVM, resulting in inconsistent task contexts in the memory hierarchy. Thus, to completely preserve the contexts of lengthy tasks, we prevent a lengthy task from being scheduled at a low voltage by forcing the scheduler to context switch the currently executed task if it is lengthy, so that the variables fetched into the CPU registers will be pushed on top of the task's stack and also preserved in NVM during power-off periods.

	\subsubsection{Two-version Data Allocation}
	In addition to local variables, all tasks are allowed to access data objects which may be shared by multiple concurrently executed tasks. The data manager maintains two respective versions (i.e., \textit{working} version and \textit{consistent} version) for each data object, where the working version (allocated in VM by default unless otherwise specified) provides high performance and energy efficiency for data access, while the \emph{consistent} version (stored permanently in NVM) provides reliability and persistency when a power failure occurs. Moreover, the data manager also allows for multiple working copies for the working version of a data object, as well as allows a \textit{temporary} copy and a \textit{persistent} copy respectively in VM and NVM for the consistent version while keeping the two copies identical at all times. This can increase the flexibility of concurrent task execution by allowing multiple tasks to simultaneously access the same data object, thus improving forward progress. All data accesses are via the three operations, namely read, write, and commit, provided by the data manager. A task can read a data object by obtaining its memory address via the read operation. To improve data access efficiency, we adopt the copy-on-write strategy for the write operation. Specifically, once a task attempts to modify a data object that has yet to be modified by the task, a working copy of the data object will be created and dedicated for the task to read and write afterward. To update the persistent copy of the consistent version in NVM, a task must perform the commit operation, and the update will be made only if the serializability condition is not violated. By using the operations to access data objects, data management can take advantage of the characteristics of hybrid memory while being transparent to tasks.
	
	\begin{figure}[h]
		\subfigure[Data access and allocation policy for non-lengthy tasks]{
			\label{fig:default}
			\includegraphics[width=0.98\columnwidth]{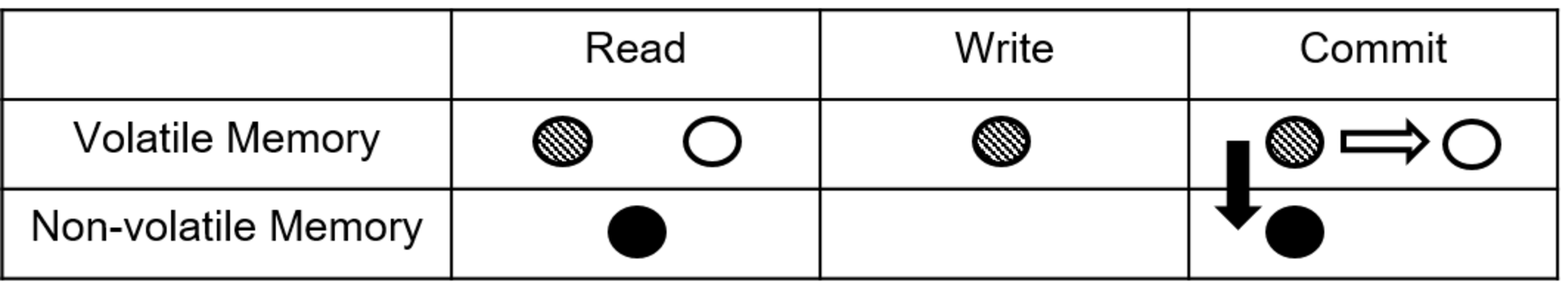}
		}
		\vspace{0.05in}
		\subfigure[Data access and allocation policy for lengthy tasks]{
			\label{fig:lengthy}
			\includegraphics[width=0.98\columnwidth]{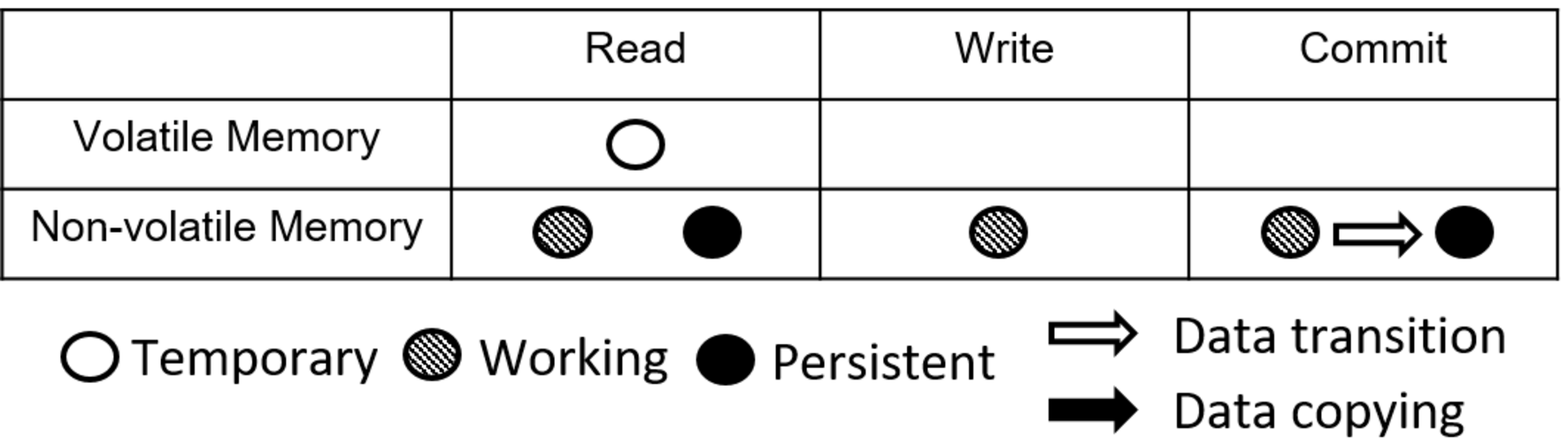}
		}
		\caption{Data copies accessed by three operations.}\label{fig:operations}
	\end{figure}

	The data manager carries out the three operations to improve forward progress while ensuring data consistency.
	Figure~\ref{fig:default} shows the data copies accessed when a \emph{non-lengthy} task invokes each of the three operations.
	Once a task invokes the read operation on a data object, the temporary copy of the data object is read by default unless the working copy dedicated for the task is available (i.e., the data object has been modified by the task). However, if the temporary copy in VM is not identical to the persistent copy in NVM (e.g., the system resumes after a power failure), the temporary copy is deemed to be \emph{invalid} and the persistent copy is read instead.
	Considering the access efficiency of writing a data object, the data manager allocates the working copy dedicated for each task in VM by default.
	A task calls the commit operation immediately before finishing its execution to update the consistent versions of those data objects modified by the task. Before updating the consistent versions, the data manager validates whether the update violates the serializability of those finished tasks. If serializability is violated, the task is aborted and rerun. Otherwise, for each data object modified by the finished task, its persistent copy in NVM is updated as the task's working copy, and the working copy becomes its temporary copy in VM. Note that a data object may have multiple working copies if it is accessed concurrently by several tasks, and the working copy left by the recently finished task in VM always transits into the temporary copy.

	The data copies accessed by a \emph{lengthy} task are slightly different from the copies accessed by a non-lengthy task, because the former is \emph{intermittently} executed in NVM while the latter is \emph{atomically} executed in VM. Figure~\ref{fig:lengthy} shows the data copies accessed when a lengthy task invokes each of the three operations. Once the lengthy task attempts to read a data object, the data copy to be read is also determined according to the default rule applied to non-lengthy tasks. The main difference is that the working copies of all data objects modified by the lengthy task will be allocated in NVM (instead of VM) to ensure that the execution progress and data modifications of lengthy tasks are consistent in NVM across power-on periods. When the lengthy task attempts to commit its modifications, the serializability is also validated to determine whether the update to consistent versions is permitted or the task should be aborted and rerun. However, if the commit operation is permitted, the working copy of each modified data object directly transits into its persistent copy in NVM (without additional memory copying from VM to NVM), but the temporary copy remains unchanged and becomes invalid in VM because it may not be identical to the persistent copy.

	To prevent data corruption due to power failure during the commit operation, the operation must be implemented to be atomically. In other words, to prevent partial updating of the consistent version in NVM, the commit operation must atomically update none or all of the modifications made by the task. To this end, we borrow an idea proposed to atomically update shadow pages from~\cite{Gray:1992:TPC:573304} and use a bit map stored in NVM to maintain the addresses of valid persistent copies. Specifically, for each data object committed by a non-lengthy task, its modification on a data object will first be made on a \textit{shadow} copy in NVM, while for each data object committed by a lengthy task, its working copy will first transit into a shadow copy. Then, those modified shadow copies and their persistent copies will be swapped by updating the bit map only after all modifications or transitions for the committed data objects are finished. Because updating the bit map only requires one CPU instruction, which is the minimum execution unit of a CPU, the commit operation is atomic and resilient against power failures. More implementation details will be discussed in Section~\ref{sssec:atomic}.

	\subsubsection{Serializability Validation}
	
	The data manager ensures serializability with a backward validation procedure which determines whether those finished tasks remain serializable if a new commit operation is performed by a task. To this end, for each finished task (whether it is lengthy or non-lengthy), the validation procedure maintains a \textit{validity time interval}\footnote{In our implementation, the unit of validity time intervals is set as one single system time tick, and the time tick is triggered every 2ms on the used Texas Instruments platform.}, in which the task can be viewed as having been executed in isolation. Moreover, each data object is also associated with a validity time interval which is updated as the validity time interval of the most recently finished task that commits the object. The validation procedure is invoked whenever a task attempts to commit its modifications to data objects. If a valid validity time interval can be derived for the task, its commit operation proceeds; otherwise, it is aborted and rerun.
	
	\begin{algorithm}[h]
		\caption{Validation Procedure}
		\label{algo:valid}
		\begin{algorithmic}[1]
			\small
			\item[\textbf{Input:}]  ${T}$, ${R = \{r_{1}}...r_{n}\}$, ${W = \{w_{1}}...w_{m}\}$,
			\State ${T.begin = 0}$;
			\State ${T.end = getcurrenttime()}$;
			\For{\texttt{i = 1 : n}}
			\State ${T.begin = max(T.begin, r_{i}.begin+1)}$;
			\If{${r_{i}.obj}$ was first modified by any finished task $\tau$ after ${r_{i}}$}
			\State ${T.end = min(T.end, \tau.begin-1)}$;
			\EndIf
			\EndFor
			\For{\texttt{i = 1 : m}}
			\State ${T.begin = max(T.begin, w_{i}.begin+1)}$;
			\If{${w_{i}.obj}$ was last modified by any finished task $\tau$ after ${w_{i}}$}
			\State ${T.begin = max(T.begin, \tau.begin+1)}$;
			\EndIf
			\EndFor
			\If{${T.begin \leq T.end}$}
			\State commit($T$);
			\Else
			\State abort($T$);
			\EndIf
		\end{algorithmic}
	\end{algorithm}

	Algorithm~\ref{algo:valid} implements the validation procedure which determines the validity time interval (${T.begin}$ to ${T.end}$) for a given task ${T}$. At runtime, the data manager records all read actions, ${R = \{r_{1}}...r_{n}\}$, made by task $T$ on the temporary or persistent copies, as well as all write actions, ${W = \{w_{1}}...w_{m}\}$, made by task $T$ on its working copies. Each read action ${r_{i}}$ records the validity time interval (${r_{i}.begin}$ to ${r_{i}.end}$) of the data object, ${r_{i}.obj}$, read by task $T$ via the $i$th read operation. Similarly, each write action ${w_{i}}$ records the validity time interval (${w_{i}.begin}$ to ${w_{i}.end}$) of the data object, ${w_{i}.obj}$, written by task $T$ via the $i$th write operation. The algorithm first initializes the validity time interval of task $T$ as the range from the outset to the current time of the system (Lines 1-2). Then, the interval shrinks according to the task's read and write actions.
	For each read action ${r_{i}}$, the beginning of the time interval of task $T$ will be pushed forward because ${T}$ can only read the data object $r_{i}.obj$ after the object is committed by another finished task at $r_{i}.begin$ (Line 4). Moreover, if the data object $r_{i}.obj$ is first committed by any finished task ${\tau}$ again after ${r_{i}}$, the end of the interval of task $T$ will be pushed backward so that task $T$ can be viewed as finished before task ${\tau}$ starts (Lines 5-6). Similarly, for each write action ${w_{i}}$, the beginning of the interval of task $T$ will be pushed forward because $T$ can only commit the data object $w_{i}.obj$ after the object is committed at $w_{i}.begin$ (Line 8). Moreover, if the data object $w_{i}.obj$ is last committed by any finished task ${\tau}$ again after ${w_{i}}$, the beginning of the interval of task $T$ will be further pushed forward so that task $T$ can be viewed as started after task ${\tau}$ commits (Lines 9-10).
	Finally, the algorithm checks whether the time interval of task $T$ is valid (Line 11). If the interval is not empty, the commit operation is performed (Line 12); otherwise, task ${T}$ will be aborted (Line 14).
	
	
	\subsubsection{Property Analysis}
	
	We now analyze the time complexity of Algorithm~\ref{algo:valid} and prove that it maintains the serializability of those finished tasks. To prove serializability, we first construct a precedence graph based on the data access operations made by finished tasks. In the precedence graph, each node represents a finished task, and an arc between two nodes indicates the precedence order between two tasks due to their data access patterns conducted on some shared data objects. Then, we show that the precedence graph is acyclic. An acyclic graph indicates that all finished tasks are \textit{conflict-serializable}~\cite{Gray:1992:TPC:573304}, in that the data access operations conducted by all finished tasks can be viewed as if these operations are conducted by the tasks executed in a serial order according to the graph.

	\begin{lemma}
		The time complexity of Algorithm~\ref{algo:valid} for validating a task ${T}$ is ${O(N + M)}$, where $N$ and $M$ respectively represent the number of data objects accessed by $T$ and the number of tasks concurrently executed with $T$.
	\end{lemma}
	
	\begin{proof}
		To validate whether the serializability is maintained after task ${T}$ is committed, Algorithm~\ref{algo:valid} examines the read and write actions made by the task. Because task ${T}$ can read at most $N$ data objects which may be concurrently modified by at most $M$ tasks, for all read actions, the algorithm needs to check the validity time intervals of at most ${N}$ data objects and at most ${M}$ concurrently executed tasks which may modify the same data objects. Moreover, for all write actions, the algorithm needs to check the validity time intervals of at most ${N}$ data objects written by task ${T}$. Therefore, the time complexity of the algorithm is ${O(N + M)}$.
	\end{proof}
	
	\begin{theorem}
		All the finished tasks validated by Algorithm~\ref{algo:valid} are \textit{conflict-serializable}.
	\end{theorem}
	
	\begin{proof}
		This theorem can be proved by showing that the precedence graph is acyclic. For ease of presentation, let ${G_{i}}$ represent the precedence graph constructed from the first ${i}$ finished tasks. We prove the theorem by mathematical induction on index ${i}$ when ${i \geq 1}$. As the induction basis, when $i=1$, the theorem is correct because only one task is committed and no cycle can be formed in the precedence graph ${G_{1}}$. For the induction hypothesis, suppose that the formula is correct for the first $k$ finished tasks, and the precedence graph ${G_{k}}$ is acyclic. We show that the formula is also correct for the first $k+1$ finished tasks.
		
		We prove that the precedence graph ${G_{k+1}}$ is acyclic by contradiction. Suppose that ${G_{k+1}}$ consists of a cycle formed by a task set $\{\tau_{1}, \tau_{2}, ..., \tau_{c}\}$ in ${G_{k}}$ and the latest committed task ${T}$. Without loss of generality, we assume that these tasks are finished in the order of $\tau_1, \tau_2, ..., \tau_c$, and $T$. A cycle formed immediately after task $T$ is committed indicates that the graph contains an arc from $T$ to ${\tau_{1}}$ and an arc from $\tau_{c}$ to $T$. Based on the precedence relationship determined by the algorithm, the arc from $T$ to ${t_{1}}$ suggests that the validity time interval of $T$ is earlier than that of ${t_{1}}$. Similarly, the arc from ${t_{c}}$ to ${T}$ suggests that the validity time interval of $T$ is later than that of ${t_{c}}$. In other words, the validity time interval of ${t_{c}}$ is earlier than that of ${t_{1}}$. However, we assume that $\tau_1, \tau_2, ...$, and $\tau_c$ are finished in order, so the validity time interval of ${t_{1}}$ should be earlier than that of ${t_{c}}$.
		This results in a contradiction and implies that ${G_{k+1}}$ will not consist of a cycle if $T$ is permitted to be committed. Therefore, we can conclude that all the finished tasks validated by Algorithm~\ref{algo:valid} are \textit{conflict-serializable}.
	\end{proof}

	\subsection{Instant System Recovery}\label{ssec:recovery}
	
	\subsubsection{Data Recovery} To enable the system to recover to a consistent state from power or task failures, the recovery handler maintains the consistency between data objects and task execution. Recall that, by atomically updating the persistent copies in NVM, the data manager prevents data objects from being partially modified even if a power failure occurs during the update. Therefore, although all temporary copies in VM are lost after a power failure, once power is restored, a persistent copy for each data object can be accessed immediately from NVM and its temporary and working copies can be recreated in VM when necessary according to the persistent copy. As a result, the recovery handler achieves instant data recovery.
	
	\subsubsection{Task Recovery}
	
	To recover tasks when power is resumed or validation fails, the recovery handler records information about whether a task is finished and its attributes (e.g., code address, name, stack size, and priority) needed to recreate the task. Note that the information is stored in a data structure in NVM. Based on the information, the recovery handler monitors unfinished tasks, recreates aborted tasks, and maintains consistency between task execution and data objects. Specifically, if the validation result for a task is not serializable, the recovery handler is notified to rerun the task by aborting and recreating the task. Similarly, if a power failure occurs, although the execution progress (e.g., data, stacks, and heaps) of non-lengthy tasks in VM and CPU registers are lost, the recovery handler simply identifies those unfinished tasks according to the data structure and recreates non-lengthy tasks to achieve instant task recovery.
	
	To prevent tasks whose execution times are too lengthy to be finished within one power-on period from being repeatedly recreated and rerun, the recovery handler detects lengthy tasks and allows the computation progress of these tasks to be preserved in NVM across multiple power-on periods. When the power is resumed, if a task has ever been recreated and still cannot be finished within the latest power-on period, the task will be created in NVM and deemed a lengthy task thereafter. Note that the classification of a task as lengthy is related to the power condition at the moment. At runtime, to successfully preserve the context of a lengthy task, including its variables fetched in the CPU registers, if the currently executed task is lengthy, the recovery handler forces the task to be switched out at a low voltage (before a potential power failure occurs). Whenever the currently executed task is switched out, the recovery handler removes all lengthy tasks from the ready queue and prevents them from being scheduled and executed until the next power-on period. Recall that once a lengthy task is detected and created, the data manager will allocate its context in NVM so that its computation progress can be preserved across power cycles. Therefore, after each power resumption, the recovery handler only needs to repeatedly add (instead of recreating) the task into the ready queue of the scheduler until the task is finished. More implementation details about low voltage detection and context switch enforcement will be respectively discussed in Sections~\ref{sssec:voltagedetect} and~\ref{sssec:contexxtswitch}.
	
	


	\subsection{Implementation Issues}\label{ssec:implemetation}
	
	Our design was integrated into FreeRTOS~\cite{FreeRTOS}, a real-time operating system supporting many kinds of commercial microcontrollers, running on an MSP-EXP430FR5994 LaunchPad~\cite{EXP430FR5994}, a Texas Instruments platform featuring 256KB FRAM (Ferroelectric Random Access Memory) and 8KB on-chip SRAM (Static Random-Access Memory). For portability across different system architectures and platforms, we integrated the proposed design into FreeRTOS while minimizing kernel code modifications. Our implementation comprises 11 files and 1460 lines of C code, among which 72 lines are scattered in 3 files belonging to the kernel\footnote{The intermittent OS is released under an open-source license and available at https://github.com/meenchen/Intermittent-OS.}. We discuss some technical issues that arise when implementing our design into FreeRTOS.
	
	\subsubsection{Operating System Integration}
	
	The data manager and the recovery handler are respectively implemented on top of the memory management and the task scheduler. FreeRTOS provides a set of APIs for program developers to create, schedule, suspend, and delete tasks, where the created tasks are executed by the scheduler using a round-robin scheduling policy. The scheduler records the status of each task by maintaining its \textit{task control block}, which keeps the task's information (e.g., code address, stack address, priority, etc.) in VM by default. Through the APIs, the recovery handler recreates tasks aborted due to validation or power failures. Moreover, if a task is deemed lengthy during system recovery, the scheduler stores the task's control blocks in NVM, keeping the statuses of lengthy tasks across power failures. To this end, we extend the kernel to record task attributes once a task is created, to allocate the context of lengthy tasks in NVM, and to count the number of context switches as the current timestamp (to avoid additional overhead caused by frequently accessing the real-time clock or timer).
	
	On the other hand, FreeRTOS supports a memory management mechanism and provides interfaces to allocate and deallocate memory space of the heap in VM by default. Through these interfaces, the data manager can manage the physical addresses of working and temporary copies for each data object and reclaim their space when they become invalid. In addition, we reserve an amount of memory space in NVM and adopt the memory management mechanism to maintain the persistent copies of data objects, the contexts of lengthy tasks, and the task attributes (i.e., code address, priority, stack size, name, and execution status) for those unfinished tasks recorded by the recovery handler. Note that, to ensure serializability of concurrent task execution, our read, write, and commit operations replace the original implementations provided by typical operating systems with concurrency control.

	\subsubsection{Early-abortion Validation} Whenever a task finishes, immediately before the modified data objects are committed, the data manager invokes the proposed algorithm to validate the serializability of read and write actions performed by the task, and if the serializability is violated, the task is aborted and rerun. To mitigate the waste of computation power on non-serializable tasks, we detect whether a task becomes non-serializable during its execution by simultaneously executing parts of Algorithm~\ref{algo:valid}. This is achieved by adding a data structure to record the validity time intervals for all tasks and data objects. According to the algorithm, the time interval of a task shrinks immediately when the task reads a data object (Lines 3-4) or when a data object read by the task is first committed by another task (Lines 5-6). Therefore, whenever the task performs a read operation, we can check whether the time interval of the task remains valid and early-abort the task once the interval becomes empty. Note that the first half of the validation procedure (i.e., Lines 3-6 of the algorithm), which has already been executed during task execution, can thus be skipped when the task attempts to commit its modifications.

	\subsubsection{Atomic Commit Operations}\label{sssec:atomic} When a task attempts to commit its modifications to persistent copies, the commit operation must atomically update none or all of the modifications to the persistent copies to prevent the consistent version in NVM from being partially updated and thus becoming inconsistent. To this end, we borrow the idea used to atomically update shadow pages from database systems~\cite{Gray:1992:TPC:573304} and employ a data structure stored in NVM to maintain the addresses of persistent copies so that all updates to the consistent version can be finalized by one single CPU instruction. A similar idea was also adopted in~\cite{ICCAD2019} to realize atomic commit operations. As shown in Figure~\ref{fig:mapstructure}, the data structure contains two address maps and a bit map. Each entry in an address map stores the physical address of a persistent copy, and each bit in the bit map is associated with a data object to indicate which address map has the valid address of its persistent copy (e.g., in the figure, the address map 0 has the valid addresses of the persistent copies for data objects 0 and 1).
	
	\begin{figure}[h]
		\centering
		\includegraphics[width=0.98\columnwidth]{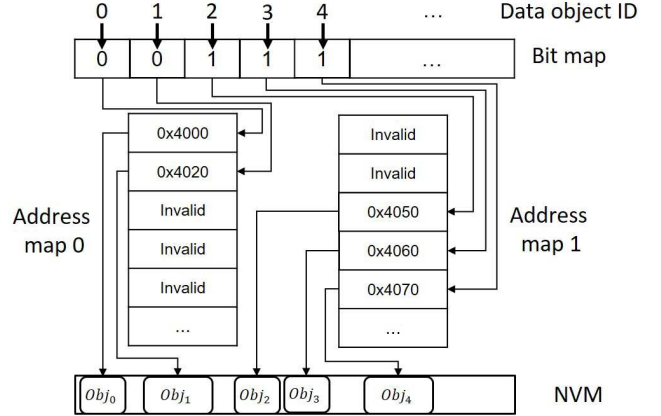}
		\caption{Data structure for addressing persistent data copies.}\label{fig:mapstructure}
		\vspace{-1em}
	\end{figure}
	
	Whenever a task finishes, for each data object modified by the task, the commit operation updates the invalid address in one of the two maps to the new address (e.g., if data object 0 is modified, the first entry in address map 1 is updated to its new address).
	Then, after the addresses of all modified data objects are updated, the commit operation simultaneously toggles the corresponding bits in the bit map (e.g.,  for data object 0, the first bit in the bit map is changed from 0 to 1).
	Because updating the bit map only requires one CPU instruction, which is the minimum execution unit of a CPU, the commit operation is atomic and resilient against power failures. However, the maximum number of bits updated by a CPU instruction is limited by the bit width of the CPU. For example, the Texas Instruments platform provides a 16-bit CPU, so at most 16 data objects can be simultaneously committed with the data structure. This should be sufficient for most applications on lightweight embedded devices. If the number of data objects modified exceeds the CPU bit width, a hierarchical bit map can be used to extend the data structure to commit a larger number of data objects.

	\subsubsection{Low Voltage Detection}\label{sssec:voltagedetect} To ensure the contexts of lengthy tasks can be successfully switched out before power failures, we use the platform’s Analog-to-digital converter (ADC) to generate interrupts when the voltage of the capacitor is lower than a threshold. After power resumption, if there are lengthy tasks running in the system, we initialize and activate the ADC to detect whether the current voltage is below the given threshold. The amount of energy stored in the capacitor can be calculated by ${\frac{1}{2}CV^2}$ (in joules), where the ${C}$ and ${V}$ respectively represent the capacitance and the current voltage of the capacitor. Based on the platform's specification, including the operating voltage ($V_{op}$), the maximum power consumption ($P$), and the context switch period of the operating system ($T_{cs}$), the threshold ($V_{th}$) can be appropriately predetermined. Once a low-voltage interrupt is triggered, to ensure that the remaining energy is sufficient to successfully switch out a lengthy task, the energy stored in the capacitor $\frac{1}{2}C(V_{th}^2 - V_{op}^2)$ must be greater than or equal to the energy required for one context switch period $P \times T_{cs}$, so the threshold can be set as $V_{th} \geq \sqrt{\frac{2PT_{cs}}{C} + {V_{op}}^2}$.

	\subsubsection{Context Switch Enforcement}\label{sssec:contexxtswitch}
	When a low voltage interrupt is triggered, the interrupt service routine notifies the recovery handler to set a low voltage flag. If the flag is set, after the currently executed task is switched out by the scheduler, the recovery handler suspends all lengthy tasks by invoking an API, namely \texttt{vTaskSuspend()}, provided by the scheduler in FreeRTOS to remove every lengthy task from the ready queue. Therefore, after the currently executed task, which could be lengthy, is successfully switched out, only non-lengthy tasks are eligible to be scheduled and the contexts of lengthy tasks will be preserved in NVM in the current power-on period. After power resumption, the recovery handler resumes all lengthy tasks by invoking another API, namely \texttt{vTaskResume()}, provided by the scheduler to put every lengthy task back into the ready queue.
	
	\subsubsection{Compatibility with Hardware-assisted Checkpointing} Our design is also compatible with a hardware-assisted checkpointing mechanism, e.g., NVP-based devices that automatically checkpoint all volatile data to NVM when a low voltage is detected~\cite{6341281}. However, a checkpointing failure would lead the system status to be rolled back to the latest successful checkpoint~\cite{Ransford2011}. As a consequence, some finished tasks could potentially be rolled back and update their modifications to the consistent version in NVM again, thus giving a rise to data inconsistency. To address this issue, if the system status is rolled back to the latest checkpoint after power consumption, the recovery handler can simply delete those finished tasks that have successfully committed their modifications based on the data structure it maintains in NVM.

	\section{Performance Evaluation} \label{sec:Evaluation}
	
	\subsection{Experimental Setup}
	\begin{table}[h]
		\centering
		\small
		\begin{tabular}{|l|p{4cm}|}
			\hline
			\multicolumn{2}{|c|}{Hardware} \\
			\hline
			MCU & 16-bit RISC EXP430FR5994 \\
			Memory & 8 KB SRAM \& 256 KB FRAM\\
			\hline
			\multicolumn{2}{|c|}{Software} \\
			\hline
			OS & FreeRTOS V9.0.0\\
			\hline
			\multicolumn{2}{|c|}{Energy harvesting management \& Power supply} \\
			\hline
			Capacitance & 200${\mu}F$\\
			Switch on/off voltage & 2.8 V/2.4V \\
			Strong power source & 3mW = 3V $\times$ 1mA\\
			Weak power source & 1.5mW = 1.5V $\times$ 1mA\\
			\hline
		\end{tabular}
		\vspace{0.05in}
		\caption{Specifications of the experimental platform.}\label{table:msp}
	\end{table}
	
	We conducted a series of experiments on the Texas Instruments platform with an energy harvesting management (EHM) module. Table~\ref{table:msp} details the specifications of the related hardware and software. The platform is powered by the EHM unit which consists of a BQ25504 low-power boost converter, a 200 $\mu F$ capacitor to store the harvested energy, and a switch to turn on (resp. off) the power supply of the platform when the voltage of the capacitor raises above 2.8V (resp. drops below 2.4V). We used a programmable power supply made by B\&K Precision to emulate the power source for the EHM. To simulate different energy harvesting sources while making the experiments reproducible, we manufactured strong (3mW = 3V$\times$1mA) and weak (1.5mW = 1.5V$\times$1mA) power traces, each of which lasted 100 seconds and was sufficient to mitigate experimental variances while reproducing the results. Neither power source was sufficient for the platform to operate continuously, repeatedly resulting in power failures and resumptions depending on the amount of energy harvested and consumed during system operation. The experimental environment is shown in Figure~\ref{fig:ExperimentEnviroment}.
	
	\setcounter{figure}{3}
	\begin{figure}[h]
		\centering
		\includegraphics[width=0.95\columnwidth]{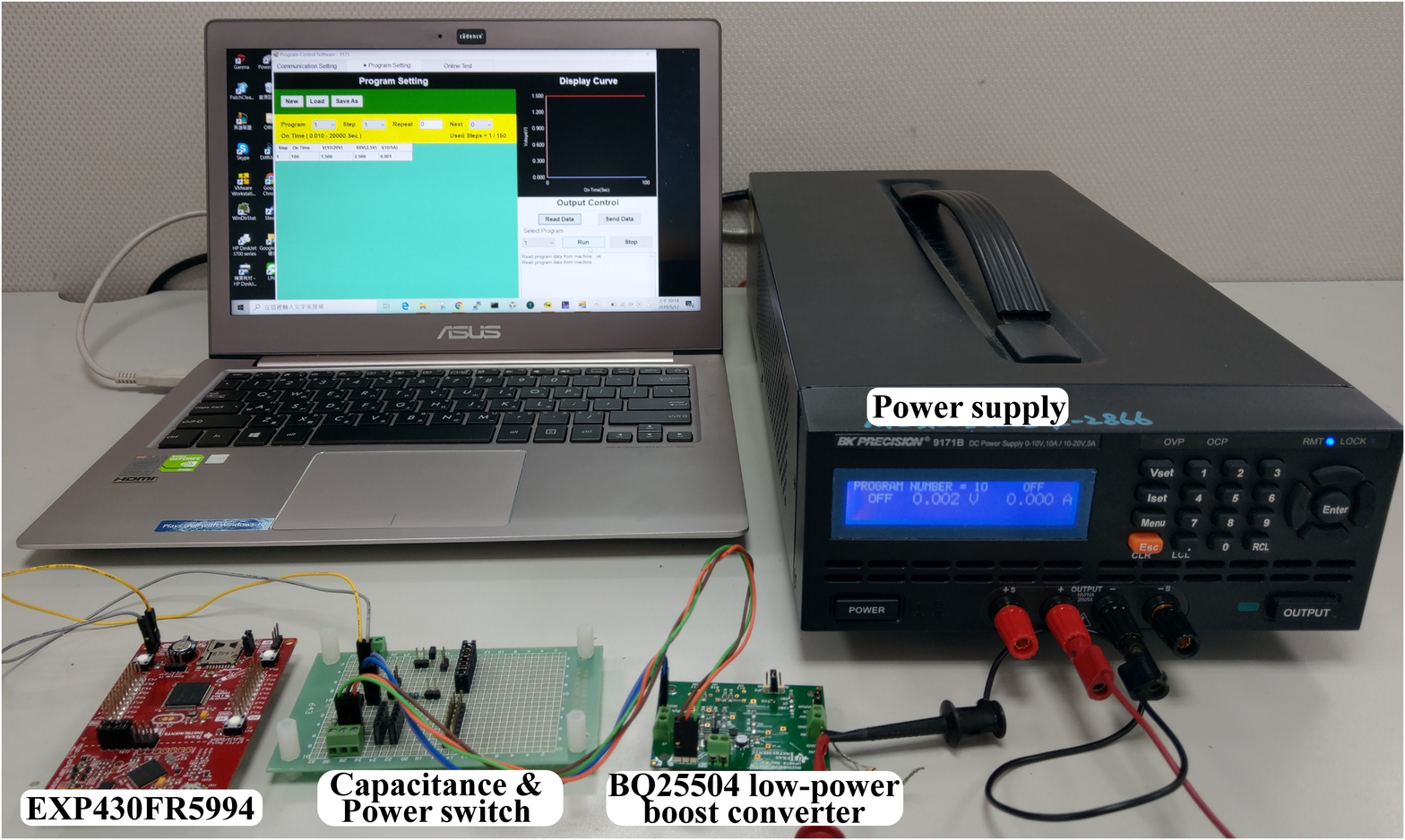}\vspace{-0.0in}
		\caption{The experimental environment.}\label{fig:ExperimentEnviroment}
	\end{figure}
	\setcounter{figure}{4}

	Given that self-powered devices typically run simple applications for data collection and processing, we ported four tasks from the benchmarks\footnote{MSP430 Competitive Benchmarking is a collection of applications used to evaluate different aspects of the microcontroller's performance.} provided by Texas Instruments and implemented one task to encrypt and transmit resultant data to an external device. Specifically, the four tasks respectively perform matrix multiplication, floating-point arithmetic, integer arithmetic, and a finite impulse response filter based on given inputs, and then commit their computation results to four respective data objects. The last task reads all the four data objects, performs SHA-256 (a secure hash algorithm) to encrypt the data, and transits the encrypted data to an external device via a universal asynchronous receiver-transmitter (UART) interface. Thus, the four data objects are read, written, and committed by the five concurrently executed tasks.
	
	To derive the low voltage threshold, $V_{th}$, which is sufficient to successfully switch out a running lengthy task before a power failure, we used a profiling tool, namely EnergyTrace Technology~\cite{energytrace} provided by Texas Instrument. Based on our measurement, the maximum power consumption $P$ is up to 5.25 mW when these tasks are concurrently executed on the platform. Alternatively, the maximum power consumption can be directly obtained from the specifications of the used platform and external modules. As discussed in Section~\ref{sssec:voltagedetect}, once a low-voltage interrupt is triggered, the remaining energy must be greater than $P \times T_{cs}$, where the context switch period $T_{cs}$ is 1 ms in FreeRTOS. Thus, the remaining energy must be greater than 5.25 ${\mu J}$ and, according to the platform specifications in Table~\ref{table:msp}, $V_{th}$ can be derived and set as $2.42$ V.

	To ensure the serializability of task execution and allow instant recovery, our design needs to validate the data access operations made by tasks and maintain data in hybrid memory at runtime. We first evaluated the overall additional costs incurred by our design, by comparing the forward progress (i.e., the number of finished tasks per second) achieved by our design and native FreeRTOS, when the device is powered with a \emph{stable} power supply. Then, we conducted breakdown analysis on the costs. We measured the time and space costs required by our design, which requires additional computation time and memory space to respectively invoke data access operations and record the task attributes. Moreover, because our design accumulates the computation progress of lengthy tasks at the cost of increased execution time and energy consumption due to NVM access latency, we measured the execution time and energy consumption of each task when its context is allocated in VM or NVM.

	To gain more insights into our design, which achieves data consistency without runtime checkpointing and system logging, we compared the performance of our design to that of the \textit{system-wise checkpointing}~\cite{6733152} and \textit{logging-based checkpointing}~\cite{databaselogging} approaches described in Section~\ref{ssec:example}, respectively denoted as \textit{SYS} and \textit{LOG}. To explore the impact of different checkpointing periods, we measured the performance achieved by SYS and LOG when they perform checkpointing frequently (i.e., every 20ms) and infrequently (i.e., every 200ms). Note that LOG and SYS adopt our validation procedure to ensure the serializability of concurrent task execution, because they originally do not consider task concurrency. All tasks were run repeatedly, and the number of finished tasks per second (i.e., forward progress) was adopted as the performance metric. Finally, to explore the runtime overheads incurred to enable intermittent computing, we measured the suspension time, the recovery time, and the data recentness achieved by \textit{SYS}, \textit{LOG}, and our design. These runtime overheads affect the forward progress and data quality when the system suffers from frequent checkpointing and recovery due to unstable power supply.

	\subsection{Experimental Results}
	
	\subsubsection{Cost measurement}
	
	Our design enables an embedded operating system to achieve serializability and data consistency at the cost of additional overheads. Figure~\ref{fig:costStable} shows that the forward progress achieved by FreeRTOS with our design integrated is reduced by 6.9\%. Note that native FreeRTOS does not guarantee the serializability of task execution, so the resultant values of data objects could be unpredictable, and the permanent version in NVM could become inconsistent after power failures. In contrast, our design provides the read, write, and commit operations to ensure serializability, while maintaining task contexts and data copies in hybrid memory to achieve data consistency. To analyze the costs, we investigated the average computation time of each operation, as well as the additional memory space required respectively by the data manager and recovery handler. We also evaluated the average execution time and energy consumption required by each task when its context is allocated in VM or NVM.

	\begin{figure}[h]
		\centering
		\includegraphics[width=0.8\columnwidth]{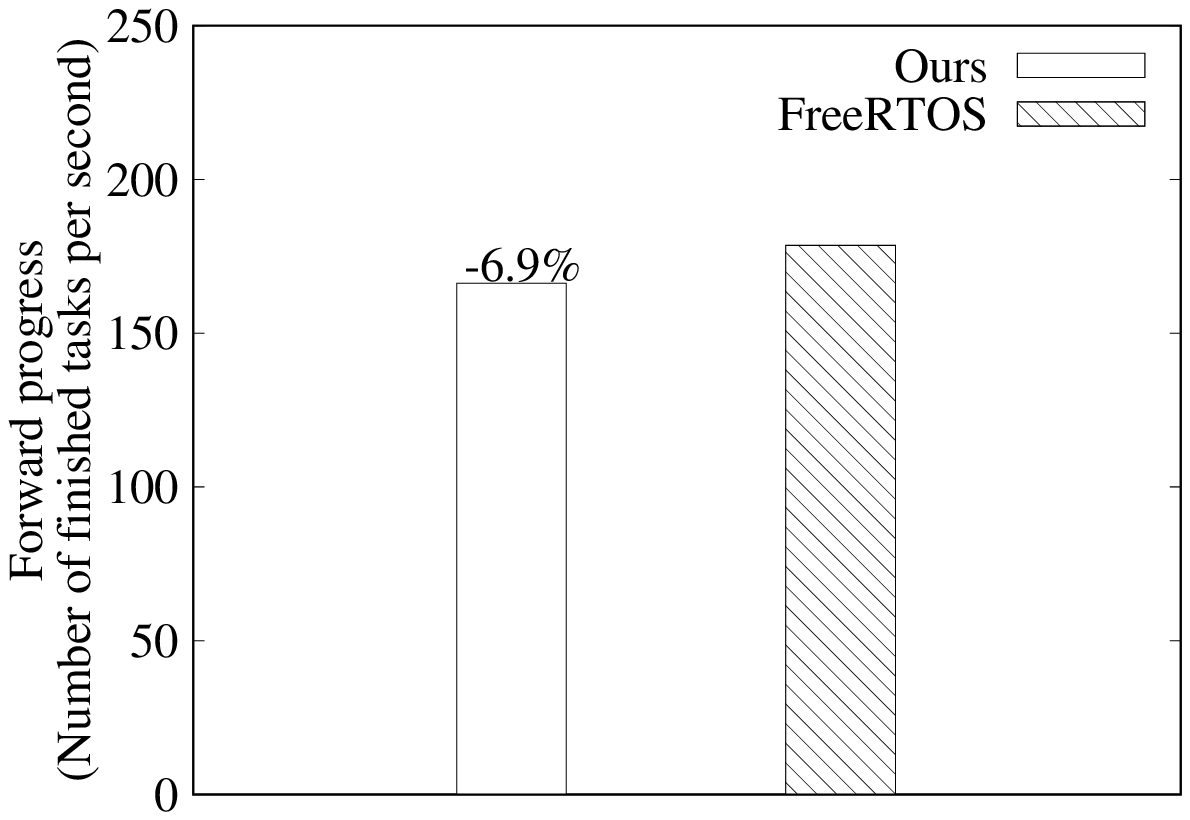}\vspace{-2em}
		\caption{Forward progress achieved by FreeRTOS with and without our design.}\label{fig:costStable}
	\end{figure}

	\begin{table}[h]
		\small
		\centering
		\begin{tabular}{|c|c|l|c|c|c|l|}
			\hline
			& \multicolumn{2}{c|}{Read} & \multicolumn{2}{c|}{Write} & \multicolumn{2}{c|}{Commit} \\ \hline
			Average execution time  & \multicolumn{2}{c|}{48 $\mu$s}   & \multicolumn{2}{c|}{65 $\mu$s}    & \multicolumn{2}{c|}{93 $\mu$s}    \\ \hline \hline
			& \multicolumn{3}{c|}{Data manager}        & \multicolumn{3}{c|}{Recovery handler}     \\ \hline
			Additional memory usage & \multicolumn{3}{c|}{5378 bytes}          & \multicolumn{3}{c|}{164 bytes}            \\ \hline
		\end{tabular}
		\vspace{0.5em}
		\caption{Execution time required by our data access operations and additional memory space used by our design.}\label{tbl:overhead}
	\end{table}

	Table~\ref{tbl:overhead} lists the average execution time of each operation, where the incurred cost in terms of average execution time is 48, 65, and 93 $\mu$s respectively for the read, write, and commit operations. The commit operation requires relatively more time than the read and write operations, because the data manager needs to validate the serializability of task execution and update the data structure which ensures the atomicity of the commit operation. By contrast, the read operation only accesses the addresses of data copies from the address maps, and the write operation only modifies working versions with copy-on-write. However, compared to the task execution time which is in a range of a few to hundreds of milliseconds, runtime overheads incurred by these operations are almost negligible. Moreover, our design uses an additional 5378 bytes and 164 bytes over the 256KB + 8KB memory space to respectively store the data structures maintained by the data manager and task attributes recorded by the recovery handler. Thus, both the time and space costs of the operations are justifiable.
	
	\begin{table}[h]
		\small
		\centering
		\begin{tabular}{|c|c|c|c|c|c|}
			\hline
			& MatMul & FIR filter  & SHA256 & Float math & Int. math\\ \hline
			VM  &  439 ms      &   336 ms  &   246 ms         &    1.89 ms        &    1.5 ms             \\ \hline
			NVM &   470 ms     &  352 ms   &   265 ms          &   1.9 ms         &    1.53 ms           \\ \hline
		\end{tabular}
		\vspace{0.2em}
		\caption{Execution time required by each task.}\label{tbl:tasktime}
		\vspace{-0.5em}
	\end{table}
	
	\begin{table}[h]
		\small
		\centering
		\begin{tabular}{|c|c|c|c|c|c|}
			\hline
			& MatMul & FIR filter  & SHA256 & Float math & Int. math \\ \hline
			VM  &  1.67 mJ           &   1.44 mJ   &   1.04 mJ         &  5.6 $\mu J$         &    4.3 $\mu J$     \\ \hline
			NVM &  2.21 mJ           &   1.56 mJ  &    1.37 mJ         &  5.7 $\mu J$          &    4.4 $\mu J$        \\ \hline
		\end{tabular}
		\vspace{0.2em}
		\caption{Energy consumption required by each task.}\label{tbl:taskenergy}
	\end{table}

	Tables~\ref{tbl:tasktime} and~\ref{tbl:taskenergy} respectively show the average execution time and energy consumption required by each task when its context is allocated in VM or NVM. Overall, the execution time of a task will be increased by between 2\% and 7\% when its context is allocated in NVM than in VM, and the energy consumed by a task will be increased by between 2\% and 32\%. The result is as expected because accessing NVM requires more energy and time than accessing VM. Consequently, when a task is deemed lengthy, the additional cost of a memory-intensive task (e.g., matrix multiplication) is higher than that of a computation-intensive task (e.g., integer arithmetic) due to frequent memory access. Note that our design allocates the context of a task in NVM only when the task is deemed lengthy to preserve its computation progress across power cycles. In our experimental settings, three tasks which respectively perform a finite impulse response filter, matrix multiplication, and SHA-256, are often deemed lengthy because they cannot be finished within a power-on period under the weak power source.

	\begin{figure}[h]
		\centering
		\captionsetup{justification=centering}
		\subfigure[Non-lengthy tasks.]{
			\label{fig:NonLengthyStrong}
			\includegraphics[width=0.95\columnwidth]{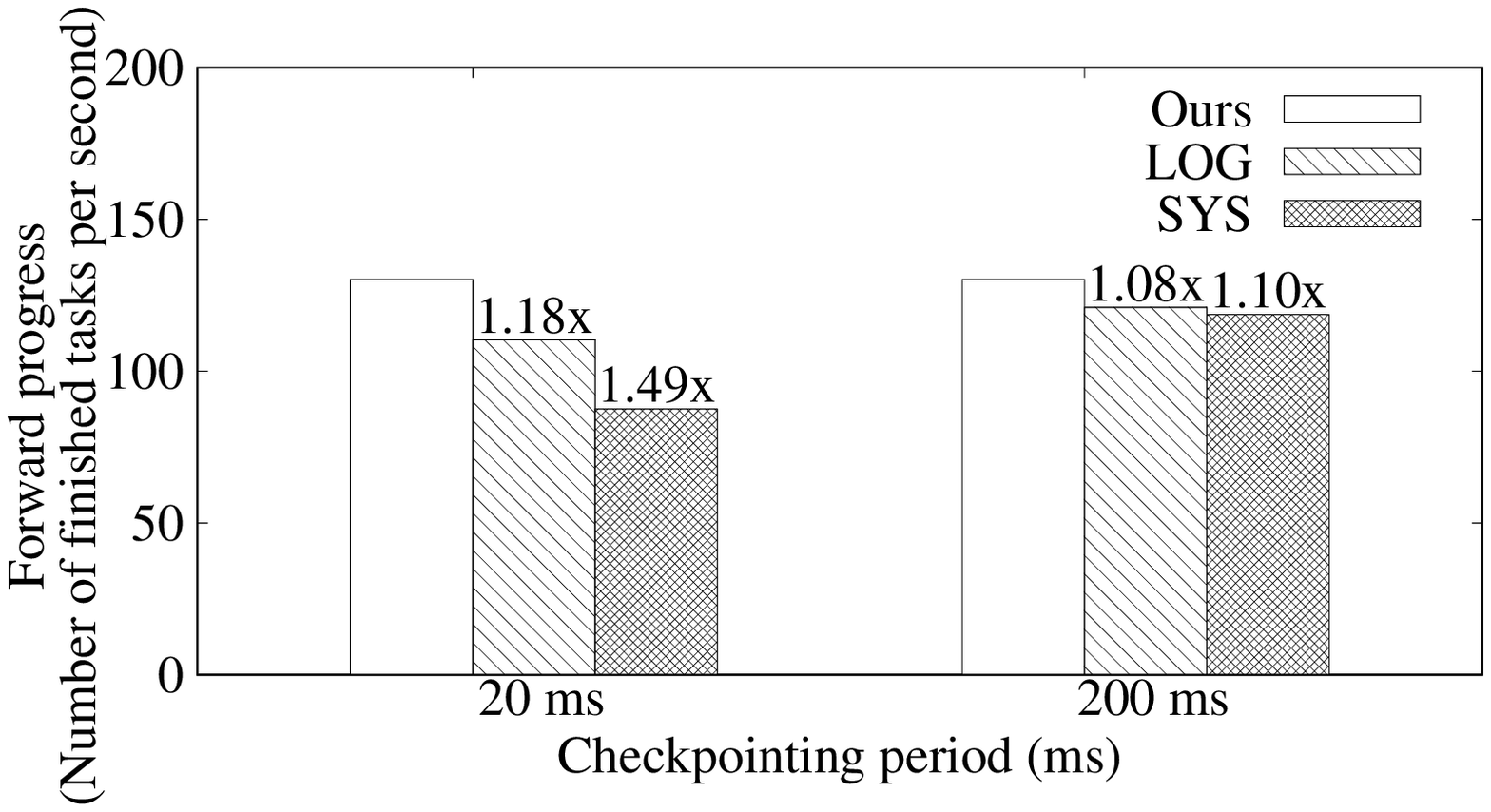}
		}
		\subfigure[Lengthy tasks.]{
			\label{fig:LengthyStrong}
			\includegraphics[width=0.95\columnwidth]{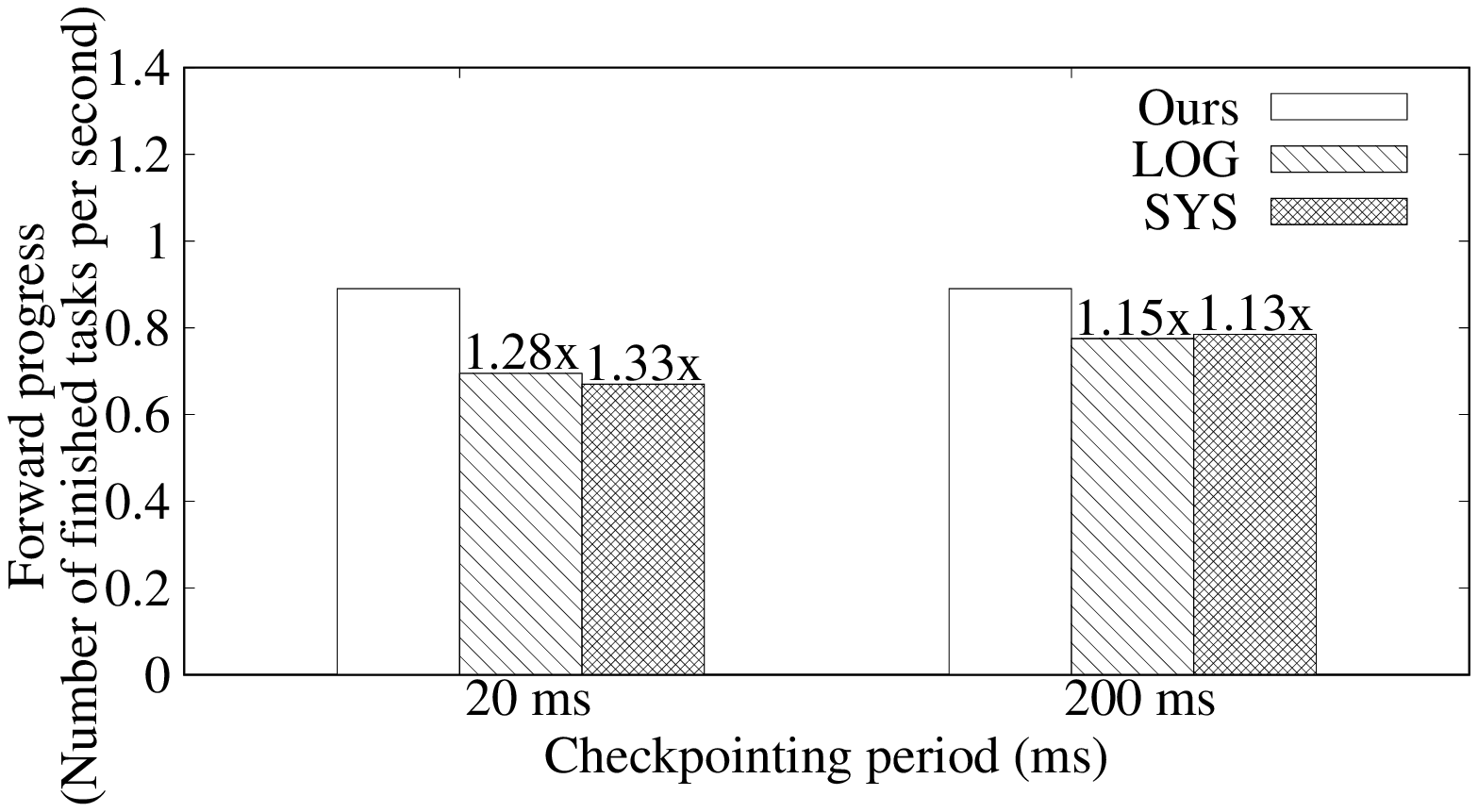}
		}
		\caption{Forward progress achieved by our design, SYS, and LOG under the strong power source.}
		\label{fig:Strong}
	\end{figure}

	\begin{figure}[h]
		\centering
		\captionsetup{justification=centering}
		\subfigure[Non-lengthy tasks.]{
			\label{fig:NonLengthyWeak}
			\includegraphics[width=0.95\columnwidth]{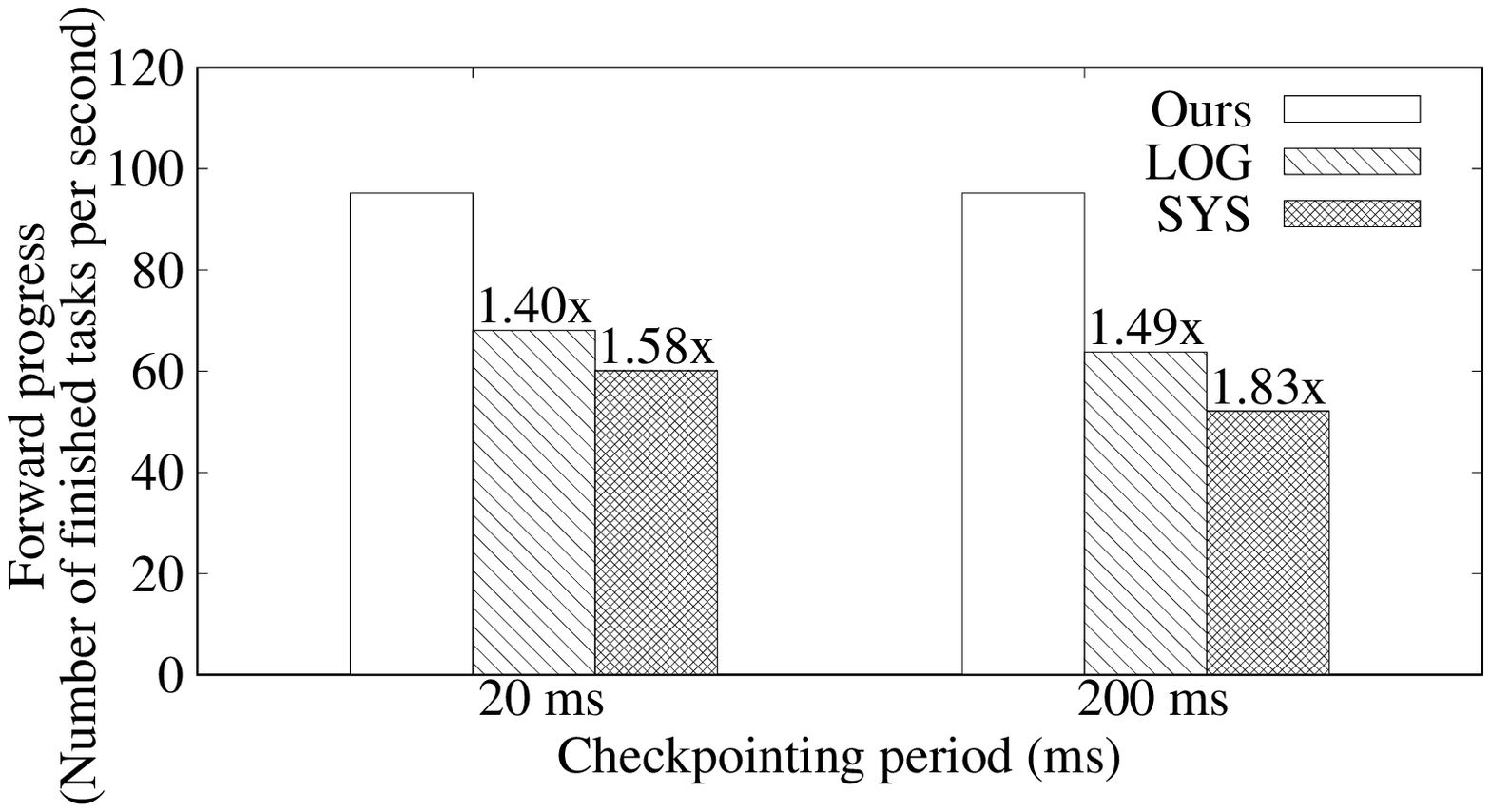}
		}
		\subfigure[Lengthy tasks.]{
			\label{fig:LengthyWeak}
			\includegraphics[width=0.95\columnwidth]{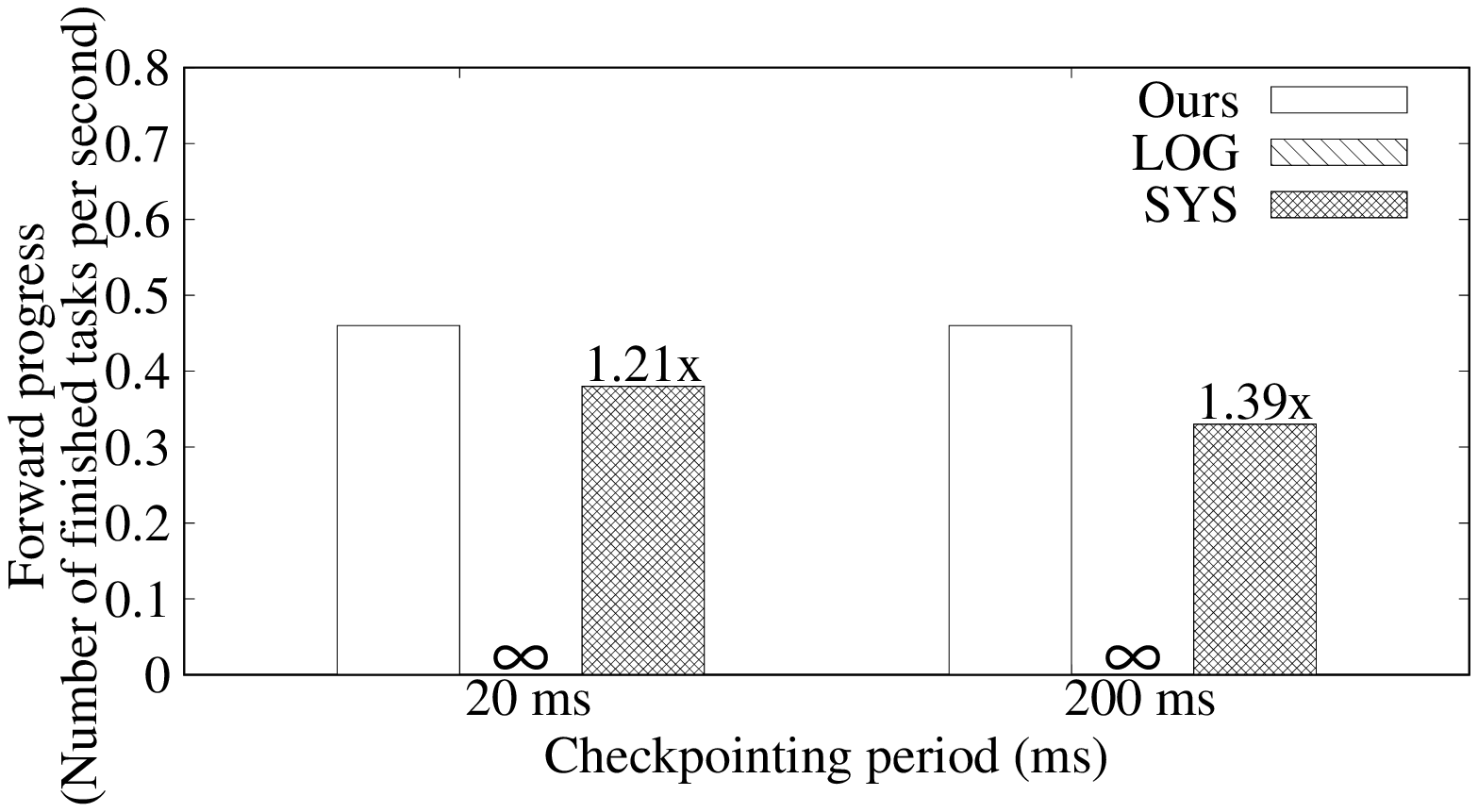}
		}
		\caption{Forward progress achieved by our design, SYS, and LOG under the weak power source.}
		\label{fig:Weak}
	\end{figure}
	
	\begin{table*}[h]
		\renewcommand\arraystretch{1}
		\centering
		\small
		\begin{tabular}{|c|c|c|c|c|c|}
			\hline
			& Ours     & SYS (20ms) & SYS (200ms)& LOG (20ms) & LOG (200ms)\\ \hline
			Suspension time (ms) & 0     & \multicolumn{2}{c|}{7.5 } & \multicolumn{2}{c|}{3.2 } \\ \hline
			Recovery time (ms) & 0.6      & \multicolumn{2}{c|}{7.6 }  & \multicolumn{2}{c|}{7 }  \\ \hline
			Data recentness (ms) & 4.7  & 10.3    &  97.8    & 10.4     &  99.6 \\ \hline
		\end{tabular}
		\vspace{0.5em}
		\caption{Average checkpoint time, recovery time, and data recentness achieved by our design, SYS, and LOG.}\label{tbl:resumption}
		\vspace{-0.5em}
	\end{table*}
	
	\subsubsection{Forward progress}
	
	Figures~\ref{fig:NonLengthyStrong} and~\ref{fig:LengthyStrong} respectively show the forward progress of non-lengthy and lengthy tasks achieved by our design, LOG, and SYS when the device is powered by the strong power source. In general, our design outperforms SYS and LOG for both long or short checkpointing periods. The forward progress achieved by our design is 1.1 to 1.49 times that achieved by SYS and is 1.08 to 1.28 times that achieved by LOG. The improved forward progress is mainly because our design eliminates the runtime overheads of snapshot checkpointing and data logging, which are respectively required by SYS and LOG to preserve forward progress and maintain data consistency. Therefore, when the short checkpointing period is adopted, our design achieves more progress improvement by eliminating the overheads frequently incurred by the checkpointing-based approaches.
	
	Figures~\ref{fig:NonLengthyWeak} and~\ref{fig:LengthyWeak} respectively show the forward progress of non-lengthy and lengthy tasks achieved by different approaches when the weak power source is adopted. For non-lengthy tasks, as shown in Figure~\ref{fig:NonLengthyWeak}, our design achieves 1.58 to 1.83 times forward progress achieved by SYS and 1.4 to 1.49 times forward progress achieved by LOG. The efficacy of our design becomes more manifest when the power supply is relatively unstable because our design enables instant system recovery, whereas SYS and LOG respectively need to restore the system snapshot from NVM to VM and traverse logs in NVM to maintain data consistency during system recovery. For lengthy tasks, as shown in Figure~\ref{fig:LengthyWeak}, LOG makes no forward progress because lengthy tasks cannot be finished within a power-on period and will be rolled back to the outset after power resumption. Compared with SYS, our design achieves 1.21 to 1.39 times more forward progress. The improved progress is because our design allocates the contexts of lengthy tasks in NVM to preserve their progress across power-off periods and switches them out before power failures, incurring less overheads compared to checkpointing the entire system snapshot from VM to NVM at runtime. Note that because a power-on period is usually much longer than the context switch period (e.g., every 1 ms in the FreeRTOS version used for our implementation) and lengthy tasks will only be switched out in a context switch period towards the end of a power-on period, lengthy tasks will still be executable in a large portion of the power-on period.
	
	Comparing Figures~\ref{fig:Strong} and~\ref{fig:Weak}, when the power supply is relatively stable, both SYS and LOG achieve more forward progress with a longer checkpointing period than with a short checkpointing period. However, when the power supply is relatively unstable, the forward progress achieved with a longer checkpointing period decreases more substantially than with a shorter checkpointing period because more uncheckpointed progress could be lost under power failure conditions. This also raises a robustness consideration because the performance of checkpointing-based approaches is highly dependent on the relationship between the checkpointing period and the power failure period.
	
	\subsubsection{Runtime overhead}

	To enable intermittent computing, the system may be suspended at runtime for checkpointing and take some time to recover after power resumption, incurring additional runtime overheads. As to the runtime overheads incurred by checkpointing and recovery, we measured the average time required for system suspension during each checkpointing, the average time required for system recovery after power resumption (i.e., when the first task can be run after power resumption), and the average time required to complete a non-lengthy or lengthy task. Moreover, the recentness of data objects after recovery (i.e., the time difference between the last data update and the recovered system) was measured to evaluate the quality of data over intermittent execution.

	As shown in Table~\ref{tbl:resumption}, compared to SYS and LOG, our design completely eliminates runtime suspension, which respectively takes 7.5 and 3.2 ms for SYS and LOG. By eliminating the time required to restore the system snapshots back from or log traversing in NVM, our design reduces the recovery time required by SYS and LOG respectively from 7.6 and 7 ms to 0.6 ms, a reduction of at least 90\%. Our design achieves a shorter recovery because, after power resumption, it simply reruns unfinished non-lengthy tasks and adds unfinished lengthy tasks into the ready queues based on their attributes maintained by the recovery handler in NVM.
	Moreover, our design significantly improves the data recentness achieved by SYS and LOG, and the improvement is more manifest as the checkpointing period increases because SYS and LOG will roll the data back to an older version after power resumption.
	
	To sum up, extensive experiments based on a prototype system running real tasks demonstrate that our design not only ensures data consistency but also outperforms checkpointing-based approaches in terms of the forward progress, the checkpoint time, the recovery time, and the data recentness after power resumption. This also suggests that our design is particularly suitable for self-powered devices which may suffer from frequent power failures.

	\section{Concluding Remarks}\label{sec:conclusion}
	We present a failure-resilient design, which employs a data manager and a recovery handler, to endow intermittent systems with concurrent task execution, data consistency without runtime suspension, instant system recovery, and stagnation-free computation. The data manager maintains serializability of concurrent task execution by controlling the data access operations conducted by tasks while ensuring data consistency by atomically committing data copies modified by finished tasks in VM to persistently consistent versions in NVM. In contrast to checkpointing-based approaches which require frequent system suspension to back up volatile data at runtime, the persistently consistent version allows the recovery handler to instantly recover the system by rerunning all unfinished tasks whose progress is lost in VM due to power failures, thereby eliminating the time required to restore the system to a previously checkpointed state. Moreover, to accumulate the progress of lengthy tasks across power cycles, the data manager allocates the data and contexts of lengthy tasks in NVM, and the recovery handler allows these tasks to instantly resume their executions based on their states preserved in NVM after power resumption. We implemented the data manager and the recovery handler on top of the memory management and the task scheduler in FreeRTOS. The results of experiments conducted on a Texas Instruments EXP430FR5994 LaunchPad show that our design significantly increases the forward progress achieved by system-wise checkpointing~\cite{6733152} and logging-based checkpointing~\cite{databaselogging} approaches. It also suggests that our design is particularly appropriate for self-powered devices suffering frequent power disruptions because the design guarantees data consistency while reducing the runtime overhead and the recovery time.
	

	To further improve forward progress, future work will seek to extend our design to consider not only data objects but also compiled program code of tasks. This will raise opportunities for, as well as challenges to, designing a memory allocation policy which considers the trade-off between the overhead incurred by copying program code from NVM to VM and performance improvements by running programs on VM.

	\section*{Acknowledgement}
	
	This work was supported in part by the Ministry of Science and Technology, Taiwan, under grant 107-2628-E-001-001-MY3. 
	\bibliographystyle{abbrv}
	\bibliography{ref}
\end{document}